\documentclass[a4paper,USenglish,cleveref,autoref,thm-restate]{lipics-v2021}

\usepackage{amsmath}
\usepackage{amssymb}
\usepackage{graphicx}
\usepackage{tikz}
\usetikzlibrary{shapes.geometric, positioning, calc}

\title{Monotone Bounded Depth Formula Complexity of Graph Homomorphism Polynomials}
\titlerunning{Monotone Bounded Depth Formula Complexity}

\author{Balagopal Komarath}{Indian Institute of Technology Gandhinagar, India}{bkomarath@rbgo.in}{https://orcid.org/0009-0008-3007-6280}{}
\author{Rohit Narayanan}{Indian Institute of Technology Gandhinagar, India}{rohitocracy@gmail.com}{https://orcid.org/0009-0007-6763-7726}{}

\authorrunning{B. Komarath and R. Narayanan}
\Copyright{Balagopal Komarath and Rohit Narayanan}



\newcommand{\coliso}{\mathsf{ColIso}}
\newcommand{\homp}{\mathsf{Hom}}
\newcommand{\cost}{\lambda}

\EventEditors{Michal Kouck\'{y} and Daniela Petri\c{s}an}
\EventNoEds{2}
\EventLongTitle{51st International Symposium on Mathematical Foundations of Computer Science (MFCS 2026)}
\EventShortTitle{MFCS 2026}
\EventAcronym{MFCS}
\EventYear{2026}
\EventDate{August 24--28, 2026}
\EventLocation{Paris, France}
\EventLogo{}
\SeriesVolume{386}
\ArticleNo{53}

\nolinenumbers

\begin{document}

\maketitle

\begin{abstract}
We introduce \emph{baggy elimination trees}, a novel graph decomposition that generalises the classical elimination trees underlying treedepth, and use them to give a complete characterisation of the monotone bounded-depth formula complexity of graph homomorphism and coloured isomorphism polynomials. Specifically, we prove that the $\Delta$-product depth monotone formula complexity of these polynomials is $\Theta\!\left(n^{\cost_\Delta(H)}\right)$, where $\cost_\Delta(H)$ is the minimum cost of a baggy elimination tree for $H$ at BET-depth~$\Delta$.

This result closes the last open case in the programme initiated by Komarath, Pandey and Rahul~\cite{KPR23} and continued by Bhargav, Chen, Curticapean and Dwivedi~\cite{BCCD25}: tight size characterisations of monotone circuit complexity (via treewidth / bounded-depth treewidth), monotone ABP complexity (via pathwidth / bounded-depth pathwidth), and monotone formula complexity (via treedepth) were already known; our theorem supplies the missing bounded-depth formula characterisation via the new notion of bounded-depth baggy-elimination-tree cost $\cost_\Delta$, completing the picture for all three models in algebraic complexity and their fixed depth variants.

As applications, for constant-degree polynomial families we derive an almost-optimal separation between monotone circuits and monotone formulas at every fixed product depth: there exists a family computable by $O(N)$-size monotone circuits of product depth $\Delta$ that requires $\Omega(N^{\Delta/2})$-size monotone formulas of the same depth (and this exponent is optimal up to a constant factor). We also prove a strict depth hierarchy: for every $\Delta \geq 1$ and every constant $k \geq 2$, there is a constant-degree family with $O(s(N))$-size monotone formulas of product depth $\Delta$ that requires $\Omega(s(N)^k)$-size monotone formulas of product depth $\Delta - 1$.
\end{abstract}

\clearpage

\section{Introduction}
The study of graph homomorphism polynomials has emerged as a surprisingly powerful lens for understanding fundamental questions in both algorithm design and algebraic complexity. As an algorithmic tool, efficient constructions for graph homomorphisms yield optimal algorithms for various pattern counting and detection problems \cite{CDM17}. In algebraic complexity, these same polynomials provide a unified framework for defining natural families that are complete for major complexity classes like $\mathsf{VP}$ and $\mathsf{VNP}$ \cite{Durand15}.

Komarath, Pandey, and Rahul \cite{KPR23} showed that the \emph{monotone} complexity, where subtractions are disallowed in computation, of graph homomorphism polynomials is characterized by various structural parameters of the pattern graph $H$. Specifically, in the unbounded-depth setting, the monotone circuit complexity is characterized by $H$'s treewidth, the monotone ABP complexity by pathwidth, and the monotone formula complexity treedepth. Therefore, we can use known separations between these graph parameters to separate corresponding algebraic computational models.

The above characterizations naturally raise the question: how does this correspondence translate to \emph{bounded-depth} computation? Recently, Bhargav, Chen, Curticapean and Dwivedi \cite{BCCD25} provided a characterization for monotone bounded-depth circuits, linking their size to a depth-restricted variant of treewidth, and for monotone bounded-depth ABPs, linking their size to a depth-restricted variant of pathwidth. However, a significant model remained uncharacterized: monotone bounded-depth formulas. This model is computationally weaker than circuits but fundamental in algebraic complexity theory. What graph-theoretic parameter, if any, governs the trade-off between depth and size for a formula computing $\homp_{H, n}$?

In this paper, we resolve this question and thereby close the last open case in the characterization programme of \cite{KPR23, BCCD25}. We introduce a graph decomposition, which we call the \emph{Baggy elimination tree}, that are generalizations of elimination trees used to define treedepth. We define two measures for a baggy elimination tree: \emph{BET-depth or the Baggy Elimination Tree depth} ($\Delta$), and \emph{cost} at BET-depth $\Delta$ ($\cost_\Delta$). We define $\cost_\Delta(H)$ as the minimum possible cost that can be achieved for a given BET-depth $\Delta$. Our main theorem proves that this graph-theoretic parameter characterizes bounded product depth monotone formula complexity for graph homomorphism polynomials.

In the Boolean setting, monotone circuits exhibit exponential gaps in depth compared to non-monotone circuits \cite{RazW92}; superpolynomial separations between monotone circuits and formulas are also well-established~\cite{KarchmerW90}, and a strict depth hierarchy is known for constant-depth monotone formulas~\cite{AjtaiG87, KPPY84, Okolnishnikova82}. In the non-monotone setting, depth hierarchies have been established for \textit{homogeneous} arithmetic circuits~\cite{KumarS17}, which is relevant in the context of homomorphism polynomials studied in this paper. Monotone formula lower bounds for computing the iterated matrix multiplication polynomial $\mathsf{IMM}_{n,d}$ (the coloured homomorphism polynomial of a path) were first proved by Shamir and Snir~\cite{ShamirSnir77}.

For general (non-monotone) arithmetic circuits, superpolynomial lower bounds against low-depth circuits have also been established~\cite{LimayeST21}.
In the monotone setting, robust separations for the monotone depth hierarchy were recently established~\cite{CGM22}. In contrast to previous existence results relying on polynomials specifically constructed only to achieve the separation, we achieve separation via a precise structural characterization for a natural family of polynomials: we show that the separation between depth $\Delta$ and $\Delta+1$ is strictly determined by the \textit{baggy elimination tree} cost of the underlying graph patterns.

The paper is organized as follows. We introduce definitions for the model and the polynomials in Section~\ref{sec:prelim}. We define baggy elimination trees, the parameter that characterizes bounded product depth monotone formula complexity, in Section~\ref{sec:baggy}. We prove the characterization in Section~\ref{sec:main}. This result provides a precise description of the complexity of a well-studied model for a large class of important polynomials in terms of a structural graph parameter, and constitutes the main contribution of our work. In fact, for any fixed graph $H$, our upper and lower bounds only differ by constant factors. For constant-degree polynomial families, we prove an almost optimal separation between monotone circuits and monotone formulas in Section~\ref{sec:cirvsfor} and a depth hierarchy theorem for monotone formulas in Section~\ref{sec:depth}. Although better lower bounds and separations are known for monotone computation (see \cite{CGM22, KPPY84}), our main contribution here is to show that these separations can be derived from easily provable separations between graph parameters, continuing the line of work in \cite{BCCD25} and \cite{KPR23}. 

\section{Preliminaries}
\label{sec:prelim}

\begin{definition}
  A polynomial over $\mathbb{Q}$ is called \emph{monotone} if all its coefficients are non-negative. An \emph{arithmetic formula} computing a polynomial in $\mathbb{Q}[x_1, \dotsc, x_n]$ is either a variable, a field constant, or $F_1 + \dotsm + F_k$, or $F_1\dotsm F_k$ where $F_i$ for $i \in [k]$ are arithmetic formulas computing polynomials in $\mathbb{Q}[x_1, \dotsc, x_n]$. The formula is called \emph{monotone} if all constants in it are non-negative.
  
  An arithmetic formula can be naturally represented as a rooted tree where the internal nodes (called gates) are labelled $+$ or $\times$ and the leaves (called input gates) are labelled by variables or constant. The \emph{size} of a formula is then the number of edges in the tree. The \emph{product depth} of the formula is the maximum number of gates labelled $\times$ over a path from the root to an input gate.
\end{definition}

The families of polynomials that we look at in this paper enumerate graph homomorphisms or colored isomorphisms. The following definitions are from \cite{KPR23}:

\begin{definition}
  For graphs $H$ and $G$, a \emph{homomorphism} from $H$ to $G$ is a function $\phi: V(H)\mapsto V(G)$ such that $\{i, j\}\in E(H)$ implies $\{\phi(i), \phi(j)\}\in E(G)$. For an edge $e = \{i, j\}$ in $H$, we use $\phi(e)$ to denote $\{\phi(i), \phi(j)\}$.
\end{definition}

\begin{definition}
  Let $H$ be a $k$-vertex graph where its vertices are labeled by $[k]$ and let $G$ be a graph where each vertex has a color in $[k]$. Then, a \emph{colored isomorphism} of $H$ in $G$ is a subgraph of $G$ isomorphic to $H$ such that all vertices in the subgraph have different colors and for each edge $\{i, j\}$ in $H$, there is an edge in the subgraph between vertices colored $i$ and $j$.
\end{definition}

\begin{definition}\label{full:def: hompoly}
  For a pattern graph $H$ on $k$ vertices, the $n$-th \emph{homomorphism polynomial} for $H$ is a polynomial on $\binom{n}{2}$ variables $x_e$ where $e = \{u, v\}$ for $u, v\in [n]$.
  
  \begin{equation*}
    \homp_{H, n} = \sum_{\phi}\prod_{e \in E(H)} x_{\phi(e)}
  \end{equation*}

  where $\phi$ ranges over all homomorphisms from $H$ to $K_n$.
\end{definition}

Computing the homomorphism polynomial is an important intermediate step in many algorithms related to finding and counting graph patterns. Instead of the homomorphism polynomial, we consider an equivalent polynomial (See Lemma~{8} in \cite{KPR23}) called the colored isomorphism polynomial which enumerates all colored isomorphisms from a pattern to a host graph where there are $n$ vertices of each color.

\begin{definition}
  For a pattern graph $H$ on $k$ vertices, the $n$-th  \emph{colored isomorphism  polynomial} for $H$ is a polynomial on $|E(H)|n^2$ variables $x_e$ where $e = \{(i, u), (j, v)\}$ for $u, v\in [n]$ and $\{i, j\}\in E(H)$.
  
  \begin{equation*}
    \coliso_{H,n} = \sum_{u_1,\dotsc,u_k \in [n]}\prod_{\{i, j\} \in E(H)} x_{\{(i, u_i), (j, u_j)\}}
  \end{equation*}
\end{definition}

We notice that the labeling of $H$ does not affect the complexity of $\coliso_{H,n}$. Given the polynomial $\coliso_{H,n}$ for some labeling of $H$ and if $\psi$ is a relabeling of $H$, then the polynomial $\coliso_{H,n}$ for the new labeling can be obtained by the substitution $x_{\{(i, u), (j, v)\}}\mapsto x_{\{(\psi(i), u), (\psi(j), v)\}}$.

Monomials of the polynomial are computed using parse trees in formulas:

\begin{definition}
  Let $g$ be a gate in a formula $F$. A \emph{parse tree} rooted at $g$ is any rooted tree which can be obtained by the following procedure:

  \begin{enumerate}
  \item The gate $g$ is the root of the parse tree.
  \item If there is a multiplication gate $g$ in the tree, include all its children in the formula as its children in the parse tree.
  \item If there is an addition gate $g$ in the tree, pick an arbitrary child of $g$ in the formula and include it in the parse tree.
  \end{enumerate}
\end{definition}

Any gate can occur at most once in any parse tree in $F$ as $F$ is a tree. Given a parse tree $T$ that contains a gate $g$, we use $T_g$ to denote the subtree of $T$ rooted at $g$. Note that we can replace $T_g$ in $T$ with any parse tree rooted at $g$ to obtain another parse tree. Similarly, if we have two parse trees $T$ and $T'$ that both contain the same multiplication gate $g$ from the formula, then we can replace any subtree of $T_g$ with the corresponding subtree of $T'_g$ to obtain another parse tree. This is because all children of $g$ in both parse trees are the same and therefore we can apply the aforementioned replacement. We call the tree obtained by removing $T_g$ from $T$ as the \emph{tree outside $T_g$ (or $g$) in $T$}.

Let $H=(V, E)$ be a graph. A vertex $v \in V(H)$ is called a \emph{pendant vertex} if its degree is one. Throughout this paper, we assume that that pattern graph $H$ has more than one edge and is connected.

\section{Baggy Elimination Trees}
\label{sec:baggy}

In this section, we introduce the most important definition in this paper and work through some examples to understand the intuition behind this definition.

\begin{definition}[Baggy Elimination Tree]
For a graph $H$, a \emph{baggy elimination tree} $T$ is a rooted tree where each node in $V(T)$ is labeled with a non-empty ``bag'' of vertices of $H$ such that every vertex of $H$ appears in exactly one bag, and the tree satisfies: if $\{u, v\} \in E(H)$, then $u$ and $v$ are either in the same bag or in bags that are in an ancestor-descendant relationship in $T$. A leaf node $t_m$ in $T$ is a \emph{core leaf} if it contains some vertex that is not pendant in $H$. Otherwise, we call the leaf \emph{non-core}.

The \emph{BET-depth or the Baggy Elimination Tree depth} of $T$ is the maximum number of nodes on any root-to-leaf path $P = (t_1, \dotsc, t_m)$, excluding the leaf node $t_m$ if $t_m$ is a non-core leaf. The cost of a single path $P$ is the sum of the cardinalities of the bags of all nodes on that path. The \emph{cost of the tree $T$} is the maximum cost over all root-to-leaf paths in $T$. The $\Delta$-BET-depth baggy elimination tree cost, denoted by $\cost_\Delta(H)$, is the cost of the minimum cost baggy elimination tree of BET-depth at most $\Delta$.
\end{definition}

To motivate the above definition, we consider monotone formulas for $\coliso_{P_7, n}$, where $P_7$ is the path on $7$ vertices. This polynomial has $n^7$ monomials. So it has an $n^7$-size monotone formula of product depth one. Komarath, Pandey, and Rahul \cite{KPR23} describe how to construct $O(n^3)$-size monotone formulas for this polynomial. Their construction has a product depth of three.

\begin{example}

Consider the following formula for $\coliso_{P_7, n}$:

{\footnotesize

\[
\sum_{i_2, i_4, i_6 \in [n]} \biggl(
\sum_{i_1 \in [n]} x_{\{(1,i_1),(2,i_2)\}} 
\sum_{i_3 \in [n]} x_{\{(2,i_2),(3,i_3)\}} \, x_{\{(3,i_3),(4,i_4)\}}
\sum_{i_5 \in [n]} x_{\{(4,i_4),(5,i_5)\}} \, x_{\{(5,i_5),(6,i_6)\}} 
\sum_{i_7 \in [n]} x_{\{(6,i_6),(7,i_7)\}} \biggr)
\]
}

This formula has size $O(n^4)$ and has product depth two. It corresponds to the baggy elimination tree of BET-depth two for $P_7$ shown in Figure~\ref{fig:path-tree}. The size of the formula is determined by the cost of the tree and the nesting of product gates increasing the product depth to two is contributed only by the core leaves in the baggy elimination tree. This is the reason for distinguishing between core and non-core leaves in the definition of baggy elimination trees.

\begin{figure}\centering
    \begin{tikzpicture}[
  core/.style={circle, draw=green!60!black, fill=green!20, thick, minimum size=18pt},
  noncore/.style={circle, draw=purple!80!black, fill=purple!10, thick, minimum size=18pt},
  corecenter/.style={ellipse, draw=orange!80!black, fill=orange!20, thick, minimum width=40pt, minimum height=25pt}
]
\node[corecenter] (core) at (0,0) {\{2,4,6\}};

\node[noncore] (a) at (-3,-2) {\{1\}};
\node[core]    (b) at (-1,-2) {\{3\}};
\node[core]    (c) at (1,-2)  {\{5\}};
\node[noncore] (d) at (3,-2)  {\{7\}};

\foreach \x in {a,b,c,d}
  \draw[thick] (core) -- (\x);

\node[noncore, label=right:{Non core leaf}] (legend1) at (2.5,-3.5) {};
\node[core, below=4pt of legend1, label=right:{Core leaf}] (legend2) {};

\end{tikzpicture}
\caption{Baggy elimination tree of BET-depth two for $P_7$}
\label{fig:path-tree}
\end{figure}
\end{example}

Baggy elimination trees are a simple generalization of elimination trees used to define treedepth. In an elimination tree, each bag has to contain exactly one vertex. Observe that any baggy elimination tree of cost $c$ can be converted into an elimination tree of cost $c$ by replacing bags with more than one vertex with a path containing those vertices. Therefore, for any $\Delta$, we have treedepth of $H$ is at most the  $\cost_\Delta(H)$ for all $H$.

\begin{remark}
    The graph that is just an edge has a formula of product depth zero that has linear size and this is optimal. Disconnected graphs can be characterized by defining \emph{baggy elimination forests} instead of trees and by defining the product depth of such forests as one plus the max of product depths of trees in the forest if there is more than one tree in the forest. In the interest of simplicity, we omit this generalization from this paper.
\end{remark}
\clearpage
\section{Characterizing Bounded Product Depth Monotone Formulas}
\label{sec:main}

Our main theorem is as follows:

\begin{theorem}
    For any connected graph $H$ on more than two vertices, the polynomial family $\coliso_{H}$ has $\Delta$-product depth monotone formula complexity of $\Theta(n^{\cost_\Delta(H)})$.
    \label{thm:main}

\end{theorem}

We prove the upper bound first.

\begin{proof}
Let $H$ be the pattern graph. For any $\Delta \geq 1$, we construct a monotone formula $\mathcal{F}$ for $\coliso_{H, n}$ with product depth $\Delta$ and size $O(n^{\cost_\Delta(H)})$. Let $T$ be a baggy elimination tree $T$ for $H$ with BET-depth $\Delta$ and cost $k = \cost_\Delta(H)$.

We build a recursive formula $\mathcal{F}$ whose structure mirrors the structure of the tree $T$. For any node $t$ in $T$, we use $X_t$ to denote all vertices in $H$ in the bag $t$ and let $A(t)$ denote the set of $t$'s proper ancestors. Let $X_{A(t)} = \bigcup_{a \in A(t)} X_a$ be the set of all vertices in its ancestors. Let $\phi_{A(t)}$ be an assignment $\phi_{A(t)}: X_{A(t)} \to [n]$. We define a formula $\mathcal{F}(t \mid \phi_{A(t)})$ that computes the polynomial for the sub-problem induced by $t$ and its descendants, given the fixed assignment $\phi_{A(t)}$ to all its ancestors.

First, we define short names for two classes of monomials for convenience. Let $\mathsf{EM}(\phi_t, X_t)$ (Edge Monomial) be the product of variables for edges \emph{within} the bag $X_t$:
\[
\mathsf{EM}(\phi_t, X_t)
= \prod_{\substack{\{u,v\}\in E(H) \\ u,v\in X_t}} x_{\{\phi_t(u), \phi_t(v)\}}
\]

EM is 1 if $X_t$ is an independent set. Now, let $\mathsf{ALM}(\phi_{A(t)}, \phi_t)$ (Ancestor Link Monomial) be the product of variables for all edges \emph{between} the current bag and \emph{an ancestor bag}:
\[
\mathsf{ALM}(\phi_{A(t)}, \phi_t)
= \prod_{\substack{\{u,v\}\in E(H) \\ u\in X_{A(t)},\, v\in X_t}} 
x_{\{\phi_{A(t)}(u), \phi_t(v)\}}
\]

If no vertex in $X-t$ is adjacent to a vertex in an ancestor bag of t, then ALM is 1. The required formula is $\mathcal{F}(r \mid \emptyset)$, where $r$ is the root node (it has no ancestors). If there are no edges contributing to the product defining EM or ALM, the corresponding monomial is equal to 1. We construct this formula inductively starting from the leaves of $t$. If $t$ is a leaf node in $T$ with ancestors $A(t)$, the formula is:
\[
    \mathcal{F}(t \mid \phi_{A(t)}) = \sum_{\phi: X_t \to [n]} \left( \mathsf{EM}(\phi, X_t) \cdot \mathsf{ALM}(\phi_{A(t)}, \phi) \right)
\]

If $t$ is an internal node with children $u_1, \dots, u_m$, its formula is:
\[
    \mathcal{F}(t \mid \phi_{A(t)}) = \sum_{\phi_t: X_t \to [n]} \left( \text{EM}(\phi_t, X_t) \cdot \text{ALM}(\phi_{A(t)}, \phi_t) \cdot \prod_{i=1}^{m} \mathcal{F}(u_i \mid \phi_{A(t)} \cup \phi_t) \right)
\]

It is easy to see that the formula is correct using an induction. We now prove that it has product depth $\Delta$ and size $O(n^k)$. Observe that each node in the tree along any root to leaf path contributes at most one to the product depth. It now suffices to show that non-core leaves of $T$ do not add to the product depth of the formula.

Consider a non-core leaf $t$ in $T$. We claim that without loss of generality we may assume $t$ contains exactly one pendant vertex. Suppose $t$ contains two pendant vertices $u$ and $v$. Since $H$ is connected and has more than one edge, there cannot be an edge between $u$ and $v$. Therefore, $u$ and $v$ are each adjacent only to some vertex in a proper ancestor bag. We may replace $t$ in $T$ with two non-core leaves $t_1$ and $t_2$ containing $u$ and $v$ respectively, without increasing the cost or the BET-depth of $T$ (non-core leaves do not contribute to product depth by definition). Repeating this argument, we may assume without loss of generality that $t = \{u\}$ for some pendant $u$ in $H$. Therefore, in the formula $\mathcal{F}(t \mid \phi_{A(t)})$, we have $\mathsf{EM}(., .) = 1$ and that $\mathsf{ALM}(., .)$ is a single variable. So there are no multiplication gates in this formula. By definition, the longest core path in $T$ has length $\Delta$. Therefore, the constructed formula $\mathcal{F}$ has a product depth of at most $\Delta$.

We now prove the size upper-bound. Observe that the fan-in of each $+$ gate at the top-level in $\mathcal{F}(t \mid .)$ is $n^k$ where $k = |X_t|$. The fan-ins of all $\times$ gates are constant (independent\footnote{These constants depend on $T$ and therefore $H$. But since we regard $H$ as a fixed pattern graph, we can absorb this cost into the $O(.)$ notation.} of $n$). Therefore, the total size of the formula is $O(n^{|X_{t_1}| + \dotsm + |X_{t_k}| = \cost_\Delta(H)})$ where $t_1, \dotsc, t_k$ is the maximum cost path in $T$.

\end{proof}

Now, we prove the lower bound. 

\begin{proof}
    
The proof is similar to the lower bound proof of Theorem~{3} in \cite{KPR23}. Given a parse tree computing a monomial of $\coliso_{H, n}$, we construct a baggy elimination tree for $H$ from the parse tree. We then show that only a few monomials can be computed using a gate that corresponds to the leaf bag in the baggy elimination tree that achieves the maximum cost. This implies a lower bound on the total number of gates.

Let $m$ be a monomial in $\coliso_{H,n}$ and let $T$ be a parse tree of $\mathcal{F}$ computing $m$ after removing all $+$ gates from the parse tree and attaching the child of $+$ gate to the $+$ gate's parent. We construct a baggy elimination tree $B$ from $T$ as follows: For each vertex $i \in V(H)$, we find all leaves in $T$ corresponding to variables that involve $i$ (e.g., $x_{\{(i, f(i)), (j, f(j))\}
}$ for all neighbors $j$ of $i$ in $H$). We put $i$ into the bag that is the Least Common Ancestor (LCA) of these leaves in $T$, denoted by $t_i = \text{LCA}(i)$. After this process, if a non-root node is empty, we remove it from $B$, attaching any children to the removed node's parent. We will see later that the root bag is either non-empty or has exactly one child after this process. If the root bag is empty and has one child, we can make the child the root. The nodes of our baggy tree $B$ then correspond to gates $g \in T$ that are an LCA for at least one vertex. We denote the bag for a node $g$ as $B(g) = \{i \in V(H) \mid \text{LCA}(i) = g\}$. The tree structure of $B$ is thus inherited from $T$.

We claim that $B$ is a valid baggy elimination tree of BET-depth $\Delta$ for $H$. For any edge $\{i, j\} \in E(H)$, in the initial tree (Refer Figure~\ref{fig:lift}, Part~(a)), the corresponding leaf in $T$ is a descendant of both $t_i$ and $t_j$, implying $t_i$ and $t_j$ must also be in an ancestor-descendant relationship. We will now prove our earlier claim about the root bag. Suppose the root bag is empty and has more than one child. Let $u$ and $v$ be two vertices of $H$ in two distinct children of the root. Observe that all vertices adjacent to $u$ or $v$ in $H$ are in the respective subtrees. So there is not path between $u$ and $v$ in $H$, which contradicts the assumption that $H$ is connected.

We claim the BET-depth of $T$ is at most $\Delta$. A non-pendant vertex $i$ of $H$ appears in multiple leaves in $T$, forcing its $t_i$ to be the LCA of distinct leaves in $B$, which must correspond to a multiplication gate. Therefore, any leaf in $B$ that corresponds to an input gate can only contain pendant vertices in $H$. So the BET-depth of $B$ is bounded by the product depth of $\mathcal{F}$.

\begin{figure}
    \centering
    \begin{tikzpicture}[
    level 1/.style={sibling distance=2.2cm},
    level 2/.style={sibling distance=1cm},
    every node/.style={draw, circle, minimum size=0.9cm, inner sep=0pt},
    orange node/.style={fill=orange!30, draw=orange!70!black, thick},
    purple node/.style={fill=purple!30, draw=purple!70!black, thick},
    green node/.style={fill=green!30, draw=green!70!black, thick},
    edge from parent/.style={draw, thick},
    label/.style={draw=none, fill=none} 
]
\centering

\node[orange node] (root1) {}
    child { node[purple node] (c1_1) {\{1,2\}}}
    child { node[orange node] (c1_2) {}
        child { node[green node] (c1_2_1) {\scriptsize\{2,3\}}} 
        child { node[green node] (c1_2_2) {\scriptsize\{3,4\}}}
    }
    child { node[orange node] (c1_3) {}
        child { node[green node] (c1_3_1) {\scriptsize\{4,5\}}} 
        child { node[green node] (c1_3_2) {\scriptsize\{5,6\}}}
    }
    child { node[purple node] (c1_4) {\{6,7\}}};
\node[label] at ( $(c1_2_1)!0.5!(c1_3_1)$ ) [yshift=-1.2cm] {a) Parse tree};


\begin{scope}[xshift=8.5cm] 
\node[orange node] (root2) {\{2,4,6\}}
    child { node[purple node] (c2_1) {\{1\}}}
    child { node[orange node] (c2_2) {\{3\}}
        child { node[green node] (c2_2_1) {}} 
        child { node[green node] (c2_2_2) {}}
    }
    child { node[orange node] (c2_3) {\{5\}}
        child { node[green node] (c2_3_1) {}} 
        child { node[green node] (c2_3_2) {}}
    }
    child { node[purple node] (c2_4) {\{7\}}};
    
\node[label] at ( $(c2_2_1)!0.5!(c2_3_1)$ ) [yshift=-1.2cm] {b) Lifting};
\end{scope}


\begin{scope}[yshift=-6cm, xshift=4.25cm] 
\node[orange node] (root3) {\{2,4,6\}}
    child { node[purple node] (c3_1) {\{1\}}}
    child { node[green node] (c3_2) {\{3\}}}
    child { node[green node] (c3_3) {\{5\}}}
    child { node[purple node] (c3_4) {\{7\}}};
    
\node[label] at ( $(c3_2)!0.5!(c3_3)$ ) [yshift=-1.2cm] {c) Baggy elimination tree};
\end{scope}

\end{tikzpicture}
\vspace{-2.1cm}
\caption{Parse tree to baggy elimination tree for $P_7$}
\label{fig:lift}
\end{figure}

Consider a monomial $m$ and a parse tree $T$ computing it in $\mathcal{F}$. After the construction of the baggy elimination tree from $T$ as detailed above, let $g$ be the gate in $T$ that corresponds to the leaf in $T$ that maximizes the cost. Let $d = \cost_\Delta(H)$. Assume without loss of generality that $1, \dotsc, d$ are the vertices of $H$ appearing in $g$ and its ancestors in the baggy elimination tree. The monomial $m$ fixes an assignment for all vertices in $H$ to $[n]^d$. Let $\phi_m(i) = u_i$ for $i \in [d]$ in this assignment. We claim that for any monomial $m'$ for which $g$ appears in a parse tree computing $m'$, we must have $\phi_{m'}(i) = u_i$ for $i \in [d]$ as well. Suppose for contradiction that $\phi_{m'}(i) = v_i \neq u_i$ for some $i \in [d]$. Let $T'$ be the parse tree for $m'$. Let $g'$ be the gate in $T$ such that $i$ is in the bag corresponding to $g'$. Since there is a unique path from $g$ to $g'$ in $\mathcal{F}$, the gate $g'$ must appear in $T'$ as well. We say that a vertex $i$ in $H$ is contained in a subtree $T$ of some parse tree if there is some input gate labelled by $x_{\{(i, .), (j, .)\}}$ in the subtree $T$. It is worth noting that the ordering is irrelevant since the variables are indexed by unordered pairs; only the incidence of the edge to i is used in the argument. We now split the proof into cases:
\begin{itemize}
    \item If $g'$ is an input gate, then it is labelled $x_{\{(i, w_i), (j, w_j)\}}$ for some $w_i, w_j \in [n]$ and $j \in V(H)$. Since $g'$ is in the parse tree for both $m$ and $m'$, it must be that $w_i = u_i$ and $w_i = v_i$, a contradiction.
    \item The gate $g'$ is a multiplication gate. The vertex $i$ occurs in at least two subtrees of $g'$ as we lifted $i$ to $g'$. We split into two cases:
    \begin{itemize}
        \item If $T'_{g'}$ contains $i$, then we use one such subtree to replace the corresponding subtree in $T$. The resulting parse tree computes a monomial that contains variables indexed by $(i, u_i)$ and $(i, v_i)$, a contradiction.
        \item If $T'_{g'}$ does not contain $i$, then $i$ must occur in the tree outside $T'_{g'}$ in $T'$. We replace $T'_{g'}$ with $T_{g'}$ in $T'$. The resulting parse tree again computes a monomial that contains variables indexed by $(i, u_i)$ and $(i, v_i)$, a contradiction.
    \end{itemize}
\end{itemize}

Now, we prove the lower bound. The polynomial $\coliso_{H, n}$ has $n^k$ monomials where $k = |V(H)|$. Any monomial can be mapped to a gate as above such that a gate in the image has at most $n^{k - \cost_\Delta(H)}$ pre-images. Therefore, there must be at least $n^{\cost_\Delta(H)}$ such gates.

\end{proof}

\section{Separating Circuits from Formulas in Bounded Depth}
\label{sec:cirvsfor}

The full $b$-ary tree of depth $\Delta$ denoted $F_{b, \Delta}$ is defined as follows: $F_{b, 1}$ is the single node tree. $F_{b, \Delta + 1}$ is a root node with $b$ subtrees each isomorphic to $F_{b, \Delta}$. We note that the treedepth of $F_{b, \Delta}$ is $\Delta$ for $b \geq 2$. It is at most $\Delta$ because the tree itself is an elimination tree. For $\Delta = 1$, the lower bound is true. For $\Delta > 1$, since $b \geq 2$, at least one $F_{b, \Delta-1}$ subtree is untouched in the root. So the depth of any elimination tree is at least $1 + (\Delta - 1) = \Delta$.

Bhargav et.~al.~\cite{BCCD25} showed that the $\Delta$-product depth monotone circuit complexity of $\coliso_{F_{b, \Delta + 1}}$ is $\Theta(n^2)$. Theorem~\ref{thm:main} and the above lower bound on treedepth shows that the $\Delta$-product depth monotone formula complexity of this polynomial family is $\Theta(n^\Delta)$. Therefore, we get the following theorem separating the power of circuits and formulas of bounded product depth.

\begin{theorem}
    For any $\Delta \geq 1$, there is a constant-degree (does not depend on $N$) polynomial family that is computable by $O(N)$-size monotone circuits of product depth $\Delta$ but requires $\Omega(N^{\Delta/2})$ monotone formulas of product depth $\Delta$, where $N$ is the number of variables in the formula.
\end{theorem}

\begin{remark}
    Observe that a product depth $\Delta$ monotone circuit of size $s$ can be converted into a product depth $\Delta$ monotone formula of size $O(s^{2\Delta})$ by simply duplicating gates as needed. Therefore, this separation is optimal up to constant factors independent of $\Delta$ in the exponent of $N$.
\end{remark}

\begin{remark}
Consider $\mathrm{ColIso}_{P_{10},n}$ and $\mathrm{ColIso}_{C_9,n}$, where
$P_{10}$ is the path on $10$ vertices and $C_9$ is the cycle on $9$ vertices. Both polynomials have degree $9$. By definition of pruned-$\Delta$-treewidth ($ptw_{\Delta}(H)$)for a graph H \cite{BCCD25}, we have $ptw_{3}(P_{10})$ and $ptw_{3}(C_9)$ to be 2, and both families admit $O(n^3)$-size monotone circuits at product depth 3.

On the other hand, Figure~\ref{fig:p10} shows a baggy elimination tree of BET-depth 3 for $P_{10}$ with $\cost_3(P_{10})=4$. By Theorem~\ref{thm:main}, $\mathrm{ColIso}_{P_{10},n}$ monotone circuits of size $\Omega(n^4)$ at product depth 3. 

We claim $\lambda_3(C_9)\geq 5$. Let $T$ be any baggy elimination tree for $C_9$ and let $r$ denote its root bag. Note that since there are no pendant vertices in $C_9$, no $T$ will have a non-core leaf. We now case-split on $|r|$.
\begin{itemize}
  \item $|r|\geq 3$: Removing the vertices of $r$ from $C_9$ leaves at least one edge component, whose contribution to the cost of $T$ is $2$ (since any path on at least one edge has treedepth $\geq 2$). Hence the maximum-cost root-to-leaf path in $T$ has cost at least $3+2=5$.
  \item $|r|=2$: The optimal choice of two vertices (see Figure~\ref{fig:C9}) leaves atleast a $P_4$, and $\cost_2(P_4) \geq 3$. Now the maximum-cost is atleast $2+3=5$.
  \item $|r|=1$: Removal of one vertex leaves a $P_8$ with $\cost_2(P_8) \geq 4$,
    giving total cost $\geq 1+4=5$.
\end{itemize}
In every case $\lambda_3(C_9)\geq 5$, so by Theorem~\ref{thm:main} the monotone formula complexity of $\mathrm{ColIso}_{C_9,n}$ at product depth $3$ is $\Omega(n^5)$.

Thus, $\mathrm{ColIso}_{C_9,n}$, a degree-$9$ polynomial family admitting $O(n^3)$-size monotone circuits, requires $\Omega(n^5)$-size monotone formulas of product depth $3$, strictly improving on the $O(n^3)$ vs.\ $\Omega(n^4)$ separation obtained from $P_{10}$ at the same degree and depth. Since our primary objective  and focus is characterization of model rather than separation,  maximize separation, we present these examples to show the potential of the framework and leave an exhaustive search for optimal separations to future work.
\end{remark}

\begin{figure}
\centering

\begin{minipage}{0.48\textwidth}
\centering
\begin{tikzpicture}[
    scale=0.75, transform shape,
    level 1/.style={sibling distance=2.4cm},
    level 2/.style={sibling distance=1.4cm},
    every node/.style={draw, circle, minimum size=0.9cm, inner sep=1pt},
    orange node/.style={fill=orange!30, draw=orange!70!black, thick},
    green node/.style={fill=green!30, draw=green!70!black, thick},
    red node/.style={fill=red!30, draw=red!70!black, thick},
    edge from parent/.style={draw, thick}
]

\node[orange node] {$\{4,7\}$}
    child {
        node[orange node] {$\{2\}$}
            child { node[red node] {$\{1\}$} }
            child { node[green node] {$\{3\}$} }
    }
    child {
        node[green node] {$\{5,6\}$}
    }
    child {
        node[orange node] {$\{9\}$}
            child { node[green node] {$\{8\}$} }
            child { node[red node] {$\{10\}$} }
    };

\end{tikzpicture}

\caption{A cost-4 BET-depth-3 baggy elimination tree for $P_{10}$}
\label{fig:p10}
\end{minipage}
\hfill
\begin{minipage}{0.48\textwidth}
\centering
\begin{tikzpicture}[
    scale=0.75, transform shape,
    level 1/.style={sibling distance=3.2cm}, 
    level 2/.style={sibling distance=1.8cm},
    every node/.style={draw, circle, minimum size=0.9cm, inner sep=1pt},
    orange node/.style={fill=orange!30, draw=orange!70!black, thick},
    green node/.style={fill=green!30, draw=green!70!black, thick},
    edge from parent/.style={draw, thick}
]

\node[orange node] {$\{1,6\}$}
    child {
        node[orange node] {$\{3,4\}$}
            child { node[green node] {$\{2\}$} }
            child { node[green node] {$\{5\}$} }
    }
    child {
        node[orange node] {$\{8\}$}
            child { node[green node] {$\{7\}$} }
            child { node[green node] {$\{9\}$} }
    };

\end{tikzpicture}

\caption{A cost-5 BET-depth-3 baggy elimination tree for $C_{9}$}
\label{fig:C9}
\end{minipage}

\end{figure}

\section{Product Depth Hierarchy for Monotone Formulas}
\label{sec:depth}

We show that $F_{b, \Delta + 2}$ has high cost for BET-depth $\Delta$ baggy elimination trees.

\begin{theorem}
    For any $b > \Delta \geq 1$, we have $\cost_{\Delta}(F_{b, \Delta + 2}) \geq b + \Delta$.
\end{theorem}

\begin{proof}
    We prove by induction on $\Delta$. For $\Delta = 1$, note that the root bag must contain all non-leaf vertices of $F_{b, 3}$ and there are $b + 1$ such vertices.

    For $\Delta > 1$, we split the proof into two cases:
    \begin{itemize}
        \item If the root bag does not contain any vertex from at least one $F_{b, \Delta + 1}$ subtree, then by the induction hypothesis, the $\Delta - 1$ BET-depth baggy elimination tree for this subtree contributes at least $b + \Delta - 1$ to the cost. The root bag contains at least one vertex. So the total cost is $b + \Delta$.

        \item Otherwise, the root bag contains at least $b$ vertices. We claim that there are no vertices in the root bag from at least one of the $b^2$ subtrees isomorphic to $F_{b, \Delta}$. Otherwise, root bag itself will contain $b^2 \geq 2b \geq b + \Delta$ vertices. But the subtree $F_{b, \Delta}$ has cost at least $\Delta$ for any product depth. Therefore, the total cost is at least $b + \Delta$.
    \end{itemize}
\end{proof}

We can now prove a depth hierarchy theorem for monotone formulas.

\begin{theorem}\label{thm:depth}
    For any $\Delta \geq 1$, for any constant $k \geq 2$, there is a constant-degree (does not depend on $N$) polynomial family that has $O(s(N))$-size monotone formulas of product depth $\Delta$ but requires $\Omega(s(N)^k)$-size monotone formulas of product depth $\Delta - 1$, where $N$ is the number of variables in the formula.
\end{theorem}

\begin{proof}
    The family is $\coliso_{F_{b, \Delta + 1}}$ with $b = (k - 1)\Delta + 1$. The corresponding colored isomorphism polynomial has $O(n^\Delta)$ size monotone formulas of product depth $\Delta$ but need $\Omega(n^{b + \Delta})$ for product depth $\Delta - 1$.
\end{proof}

\begin{remark}
    Observe that any polynomial family where degree is bounded by constant $d$ has product depth one monotone formulas of size $O(n^d)$. Therefore, polynomial separations as above are the best one could hope for. However, it is possible that the above separations can be achieved using polynomials with lower (constant) degree.

    We note that the degree of $\coliso_{F_{b, \Delta+1}}$ equals $|E(F_{b, \Delta+1})|$, which grows with $b$ and $\Delta$. Since $b = (k-1)\Delta + 1$ for given constants $k$ and $\Delta$, this degree is a constant with respect to the input size $N$ (it does not grow with $N$), but it does depend on $k$ and $\Delta$. Concretely, the degree is $b \cdot \frac{b^{\Delta+1} - 1}{b - 1} \leq b^{\Delta+2}$, which is a constant for fixed $k$ and $\Delta$. Thus, for every fixed pair $(k, \Delta)$, Theorem~\ref{thm:depth} exhibits a polynomial family of constant degree (with respect to $N$) achieving the stated separation.
\end{remark}

\bibliography{references}

\appendix
\section{Examples: Formula Construction for $H = P_7$}
Let $H$ be the path graph $v_1 - v_2 - v_3 - v_4 - v_5 - v_6 - v_7$. The pendant vertices are $\{v_1\}$ and $\{v_7\}$. We use $i_j$ to denote the assignment $\phi(v_j)$. We now construct a formula from a baggy elimination tree (Refer Section~\ref{sec:main}).

\subsection{Example 1: Baggy elimination tree with BET-depth = 1}

Consider the trivial tree $T_1$ from Figure~\ref{fig:depth1} with BET-depth = 1 and $\cost = 7$. The tree $T_1$ has a single root node $R$, whose bag is $X_R = \{v_1, v_2, v_3, v_4, v_5, v_6, v_7\}$. There are no leaves apart from the root itself. 

\begin{figure}[htbp]
    \centering
    \begin{tikzpicture}
        \node[draw=orange!90!black, thick, fill=orange!10, ellipse, inner sep=8pt]
        at (0,0) {\{1,2,3,4,5,6,7\}};
    \end{tikzpicture}
    \caption{Baggy elimination tree for $P_7$ with BET-depth 1 and $\cost$ = 7}
    \label{fig:depth1}
\end{figure}

The corresponding formula is $\mathcal{F} = \mathcal{F}(R \mid \emptyset)$. Since the ancestor set $A(R)$ is empty, we apply the Base Case rule to the root node $R$. The root $R$ is a \emph{core leaf} because its bag is not a single pendant vertex. Therefore,
\[
    \mathcal{F} = \mathcal{F}(R \mid \emptyset) = \sum_{\phi_R: X_R \to [n]} \left( \text{EM}(\phi_R, X_R) \cdot \text{ALM}(\emptyset, \phi_R) \right).
\]
Here, $\text{ALM}(\emptyset, \phi_R)$ corresponds to the ancestor link monomial. Since there are no ancestors, it contributes a factor of 1. The term $\text{EM}(\phi_R, X_R)$ represents the edge monomial for the bag $X_R$, which contains all seven vertices of $P_7$. Hence, the monomial is the product of variables corresponding to all six edges in the path:
\[
    x_{i_1, i_2} \cdot x_{i_2, i_3} \cdot x_{i_3, i_4} \cdot x_{i_4, i_5} \cdot x_{i_5, i_6} \cdot x_{i_6, i_7}.
\]
Substituting this in gives the original $\Sigma_1 \to \Pi_1$ formula:
\[
    \mathcal{F} = \underbrace{\sum_{i_1, \dots, i_7 \in [n]}}_{\Sigma_1} 
    \underbrace{\left( x_{i_1, i_2} \cdot x_{i_2, i_3} \cdot x_{i_3, i_4} \cdot x_{i_4, i_5} \cdot x_{i_5, i_6} \cdot x_{i_6, i_7} \right)}_{\Pi_1}.
\]
Thus, the formula has product depth $1$. The total size of the formula is $O(n^\cost) = O(n^7)$, which also agrees with the cost $\cost = 7$.

\subsection{Example 2: Baggy elimination tree with BET-depth = 2}

We use the optimal tree $T_2$ from Figure~\ref{fig:path-tree} with BET-depth = 2 and $\cost = 4$. The root $R$ is $X_R = \{v_2, v_4, v_6\}$, and the leaves are $L_1 = \{v_1\}, L_3 = \{v_3\}, L_5 = \{v_5\}, L_7 = \{v_7\}$.

The formula is $\mathcal{F} = \mathcal{F}(R \mid \emptyset)$. Let $\phi_R$ be the assignment $\phi_R(v_j) = i_j$ for $j \in \{2,4,6\}$.
\[
    \mathcal{F}(R) = \sum_{\phi_R: X_R \to [n]} \left( \text{EM}(\phi_R, X_R) \cdot \text{ALM}(\emptyset, \phi_R) \cdot \prod_{j \in \{1,3,5,7\}} \mathcal{F}(L_j \mid \phi_R) \right)
\]
The term $\text{EM}(\phi_R, X_R)$ corresponds to the bag $X_R = \{v_2, v_4, v_6\}$, which is an independent set, so the monomial equals 1. Similarly, $\text{ALM}(\emptyset, \phi_R)$ has no ancestors, and its monomial also equals 1. This gives the $\Sigma_1 \to \Pi_1$ layers:
\[
    \mathcal{F} = \underbrace{ \sum_{i_2, i_4, i_6 \in [n]} }_{\Sigma_1} \underbrace{ \left( \mathcal{F}(L_1 \mid \phi_R) \cdot \mathcal{F}(L_3 \mid \phi_R) \cdot \mathcal{F}(L_5 \mid \phi_R) \cdot \mathcal{F}(L_7 \mid \phi_R) \right) }_{\Pi_1}
\]

We now apply the Base Case rule to each of the four leaf formulas. The ancestor set for all leaves is $A(L_j) = \{R\}$, so $\phi_{A(L_j)} = \phi_R$.

For $L_3 = \{v_3\}$ (Core Leaf),
\[
    \mathcal{F}(L_3 \mid \phi_R) = \sum_{\phi_3: X_3 \to [n]} \left( \text{EM}(\phi_3, X_3) \cdot \text{ALM}(\phi_R, \phi_3) \right)
\]
Here, $\text{EM}$ for $X_3 = \{v_3\}$ is 1. The $\text{ALM}$ term links $X_3 = \{v_3\}$ to $X_{A(L_3)} = X_R = \{v_2, v_4, v_6\}$. The edges are $(v_2, v_3)$ and $(v_3, v_4)$, giving the monomial $x_{\phi_R(v_2), \phi_3(v_3)} \cdot x_{\phi_3(v_3), \phi_R(v_4)} = x_{i_2, i_3} \cdot x_{i_3, i_4}$. Thus,
\[
    \mathcal{F}(L_3 \mid \phi_R) = \underbrace{ \sum_{i_3 \in [n]} }_{\Sigma_2} \underbrace{ (x_{i_2, i_3} \cdot x_{i_3, i_4}) }_{\Pi_2}.
\]

For $L_1 = \{v_1\}$ (Pendant Leaf),
\[
    \mathcal{F}(L_1 \mid \phi_R) = \sum_{\phi_1: X_1 \to [n]} \left( \text{EM}(\phi_1, X_1) \cdot \text{ALM}(\phi_R, \phi_1) \right)
\]
The $\text{ALM}$ term links $X_1 = \{v_1\}$ to $X_R = \{v_2, v_4, v_6\}$, and the only edge is $(v_1, v_2)$, giving monomial $x_{\phi_1(v_1), \phi_R(v_2)} = x_{i_1, i_2}$. Thus, $\mathcal{F}(L_1 \mid \phi_R) = \sum_{i_1 \in [n]} (x_{i_1, i_2})$. Similarly, $\mathcal{F}(L_5) = \sum_{i_5} x_{i_4, i_5}x_{i_5, i_6}$ and $\mathcal{F}(L_7) = \sum_{i_7} x_{i_6, i_7}$.

The product depth is 2. The size is $O(n^\cost) = O(n^4)$.

\subsection{Example 3: Baggy elimination tree with BET-depth = 3}

Consider the optimal tree $T_3$ from Figure~\ref{fig:depth3} with BET-depth 3 and $\cost = 3$. The root $R$ is $X_R = \{v_4\}$. The internal nodes are $t_1 = \{v_2\}$ (child of $R$) and $t_2 = \{v_6\}$ (child of $R$). The leaves are $L_1 = \{v_1\}, L_3 = \{v_3\}$ (children of $t_1$) and $L_5 = \{v_5\}, L_7 = \{v_7\}$ (children of $t_2$).

\begin{figure}[htbp]
    \centering
    \begin{tikzpicture}[
        every node/.style={circle, draw, thick, minimum size=18pt, inner sep=4pt},
        level 1/.style={sibling distance=40mm},
        level 2/.style={sibling distance=20mm},
        edge from parent/.style={draw, thick}
    ]
    \node[fill=orange!10, draw=orange!80!black]{\{4\}}
        child {node[fill=orange!10, draw=orange!80!black]{\{2\}}
            child {node[fill=purple!10, draw=purple!80!black]{\{1\}}}
            child {node[fill=green!10, draw=green!70!black]{\{3\}}}
        }
        child {node[fill=orange!10, draw=orange!80!black]{\{6\}}
            child {node[fill=green!10, draw=green!70!black]{\{5\}}}
            child {node[fill=purple!10, draw=purple!80!black]{\{7\}}}
        };
    \end{tikzpicture}
    \caption{Baggy elimination tree for $P_7$ with BET-depth 3 and $\cost$ = 3}
    \label{fig:depth3}
\end{figure}

The root formula is $\mathcal{F} = \mathcal{F}(R \mid \emptyset)$. Let $\phi_R(v_4) = i_4$.
\[
    \mathcal{F}(R) = \sum_{\phi_R: X_R \to [n]} \left( \text{EM}(\phi_R, X_R) \cdot 1 \cdot \mathcal{F}(t_1 \mid \phi_R) \cdot \mathcal{F}(t_2 \mid \phi_R) \right)
\]
Thus, $\mathcal{F} = \underbrace{ \sum_{i_4 \in [n]} }_{\Sigma_1} \underbrace{ \left( \mathcal{F}(t_1 \mid \phi_R) \cdot \mathcal{F}(t_2 \mid \phi_R) \right) }_{\Pi_1}$.

For the internal node formula, consider $\mathcal{F}(t_1 \mid \phi_R)$. The ancestor set is $A(t_1) = \{R\}$, so $\phi_{A(t_1)} = \phi_R$. Let $\phi_1(v_2) = i_2$. Then,
\[
    \mathcal{F}(t_1 \mid \phi_R) = \sum_{\phi_1: X_1 \to [n]} \left( \text{EM}(\phi_1, X_1) \cdot \text{ALM}(\phi_R, \phi_1) \cdot \mathcal{F}(L_1 \mid \phi_R \cup \phi_1) \cdot \mathcal{F}(L_3 \mid \phi_R \cup \phi_1) \right)
\]
For $X_1 = \{v_2\}$, $\text{EM} = 1$, and $\text{ALM}$ links $X_1 = \{v_2\}$ to $X_R = \{v_4\}$, but since there is no edge between them, the monomial is 1. Thus,
\[
    \mathcal{F}(t_1 \mid \phi_R) = \underbrace{ \sum_{i_2 \in [n]} }_{\Sigma_2} \underbrace{ \left( \mathcal{F}(L_1 \mid \phi_R \cup \phi_1) \cdot \mathcal{F}(L_3 \mid \phi_R \cup \phi_1) \right) }_{\Pi_2}.
\]

Now, consider the leaf formula for $L_3$. The ancestor set is $A(L_3) = \{R, t_1\}$, and the ancestor assignment is $\phi_{A(L_3)}$ which maps $v_4 \to i_4$ and $v_2 \to i_2$.
\[
    \mathcal{F}(L_3 \mid \phi_{A(L_3)}) = \sum_{\phi_3: X_3 \to [n]} \left( \text{EM}(\phi_3, X_3) \cdot \text{ALM}(\phi_{A(L_3)}, \phi_3) \right)
\]
Here, $\text{EM} = 1$ for $X_3 = \{v_3\}$. The $\text{ALM}$ term links $X_3 = \{v_3\}$ to $X_{A(L_3)} = \{v_2, v_4\}$. The edges are $(v_2, v_3)$ and $(v_3, v_4)$, giving monomial $x_{\phi_{A(L_3)}(v_2), \phi_3(v_3)} \cdot x_{\phi_3(v_3), \phi_{A(L_3)}(v_4)} = x_{i_2, i_3} \cdot x_{i_3, i_4}$. Thus,
\[
    \mathcal{F}(L_3 \mid \phi_{A(L_3)}) = \underbrace{ \sum_{i_3 \in [n]} }_{\Sigma_3} \underbrace{ (x_{i_2, i_3} \cdot x_{i_3, i_4}) }_{\Pi_3}.
\]

The size is $O(n^{\cost}) = O(n^3)$, matching the cost $\cost = 3$.

\end{document}